\documentclass[aps, prx,twocolumn, showpacs,superscriptaddress,10pt,floatfix]{revtex4-1}

\usepackage{graphicx}
\usepackage{grffile}
\usepackage{amsthm}
\usepackage{epsfig}
\usepackage{amsmath,amssymb,amsfonts,mathtools,bbm,MnSymbol,mathrsfs,xfrac}
\usepackage{longtable}
\usepackage{tabularx}
\usepackage{colortbl}
\usepackage[table]{xcolor}
\usepackage[breaklinks=true,colorlinks=true,linkcolor=blue,urlcolor=blue,citecolor=blue]{hyperref}
\usepackage{comment}
\usepackage{color}
\usepackage[makeroom]{cancel}
\usepackage{tikz}
\usetikzlibrary{positioning}
\usetikzlibrary{arrows.meta}

\usepackage{soul}
\pdfoutput=1

\newcommand{\ket}[1]{|#1\rangle}
\newcommand{\bra}[1]{\langle#1|}

\newcommand{\pro}[2]{|#1\rangle\langle#2|}
\newcommand{\mean}[1]{\langle#1\rangle}
\newcommand{\abs}[1]{|#1|}

\newcommand{\tr}{\mathrm{tr}}
\renewcommand{\th}{{\mathrm{th}}}

\newcommand{\norm}[1]{\left\lvert\left\lvert#1\right\rvert\right\rvert}

\newcommand{\R}{{\rho}}

\newcommand{\I}{{I}}

\newcommand{\C}{{\mathcal{C}}}

\renewcommand{\P}{{P}}

\newcommand{\HS}{\mathcal{H}}

\newcommand{\ham}{{H}}

\newcommand{\cg}{{\mathrm{cg}}}

\newcommand{\vN}{{\mathrm{vN}}}
\newcommand{\floor}[1]{{\left\lfloor #1 \right\rfloor}}

\definecolor{mygray}{gray}{0.6}

\theoremstyle{definition}

\newtheorem{theorem}{Theorem}
\newtheorem{lemma}{Lemma}

\definecolor{dfcol}{cmyk}{1, 0.2108, 0.13, 0.3}
\newcommand{\df}[1]{\ifthenelse{\boolean{}}{\textcolor{dfcol}{[{\bf DF}: #1]}}{}}

\usepackage{physics}
\usepackage[page]{appendix}

\begin{document}


\title{Work extraction from unknown quantum sources}

\author{Dominik \v{S}afr\'{a}nek}
\email{dsafranekibs@gmail.com}
\affiliation{Center for Theoretical Physics of Complex Systems, Institute for Basic Science (IBS), Daejeon - 34126, Korea}

\author{Dario Rosa}
\email{dario\_rosa@ibs.re.kr}
\affiliation{Center for Theoretical Physics of Complex Systems, Institute for Basic Science (IBS), Daejeon - 34126, Korea}
\affiliation{Basic Science Program, Korea University of Science and Technology (UST), Daejeon - 34113, Korea}

\author{Felix C. Binder}
\email{quantum@felix-binder.net}
\affiliation{School of Physics, Trinity College Dublin, Dublin 2, Ireland}

\date{\today}

\begin{abstract}
Energy extraction is a central task in thermodynamics. In quantum physics, ergotropy measures the amount of work extractable under cyclic Hamiltonian control. As its full extraction requires perfect knowledge of the initial state, however, it does not characterize the work value of unknown or untrusted quantum sources. Fully characterizing such sources would require quantum tomography, which is prohibitively costly in experiments due to the exponential growth of required measurements and operational limitations. Here, we therefore derive a new notion of ergotropy applicable when nothing is known about the quantum states produced by the source, apart from what can be learned by performing only a single type of coarse-grained measurement. We find that in this case the extracted work is defined by the Boltzmann and observational entropy, in cases where the measurement outcomes are, or are not, used in the work extraction, respectively. This notion of ergotropy represents a realistic measure of extractable work, which can be used as the relevant figure of merit to characterize a quantum battery. 
\end{abstract}

\maketitle

Efficient energy extraction is a key quest for living beings and modern technology alike. In recent years the advent of quantum technology has spurred the study of energy sources beyond the classical realm and the emerging field of quantum thermodynamics~\cite{gemmer2009quantum,kosloff2013quantum,goold2016therole,vinjanampathy2016quantum,binder2018thermodynamics,deffner2019quantum} has investigated the role of quantum features in this task. At the same time while modern quantum technology already finds applications in secure communication~\cite{minder2019experimental,Pirandola2020cryptographyreview,chen2021integrated}, sensing~\cite{accellerometer2018,menoret2018gravity,ligoquantum2019}, and computing~\cite{supremacy2019,chinasupremacy2020,ibmsupremacy2021,timecrystalmanybody2021,googletimecrystal2022} these devices need to be powered, conceivably with non-equilibrium, quantum sources of energy. An example are recently experimentally-demonstrated~\cite{quach2022superabsorption,hu2021optimal} quantum batteries~\cite{alicki2013entanglement,Campaioli2018a,bhattacharjee2021quantum, shi2022entanglement}, which offer a significant quantum advantage in charging power~\cite{Binder2015a,campaioli2017enhancing, ferraro2018high-power, rossini2020quantum, gyhm2022quantum}. Generally, if one wants to make use of energy from an energy source, the first step is to characterize it. In the quantum regime, the energetic potential of the source is given by the quantum state it produces and measured by the Hamiltonian. 
Energy can be extracted by performing operations that transform this state into a state of lower energy, and collecting the surplus in the process.

Here, we consider a quantifier of work potential in the quantum regime called ergotropy $W$ which equals the amount of energy extractable from a known quantum state $\rho$ under the application of cyclic Hamiltonian control $H(t)$ (where $H(t)=0$ for $t<0$ and $t>T$ for protocol duration $T$)~\cite{Allahverdyan2004a}. Since the resulting overall unitary evolution $U$ is reversible, no entropy or heat is produced and the energy change $\tr[H(0)(\rho-U\rho U^\dag)]$ exclusively manifests as work.

Ergotropy has been widely studied and measured in experiments~\cite{VonLindenfels2019,VanHorne2020} where it quantifies the energy deposited onto a quantum load. However, a conceptual hurdle remains: the assumption of perfect knowledge of the state from which energy is extracted. In practice, the energy source may be unknown or uncharacterized and prohibit such idealized energy extraction. To fully characterize it, one would require complete state tomography on a large number of identically-prepared states before the actual work extraction procedure, involving measurements in a number of non-commuting bases~\cite{dariano2003quantumtomography,tomographyexperimental2010,toninelli2019tomographyreview}. In many-body systems, which constitute quantum batteries, this number is enormous and thus these measurements are practically unrealistic. This is also the reason why entanglement entropy is difficult to measure, with a few exceptions in small-dimensional systems~\cite{lanyon2017efficient,sackett2000experimental}. In many-body systems, only the second order R\'enyi entropy has been measured instead~\cite{kaufman2016,brydges2019probing,islam2015measuring,su2022observation}. In many experiments only limited types of measurements can be performed~\cite{hofferberth2007non,polkovnikov2011colloquium,trotzky2012probing,schreiber2015observation,kaufman2016,bernien2017probing}.
Further, the required measurements can be prohibitively costly, may require long~\cite{strasberg2022long} or even infinite time~\cite{busch1990energy}, and may in fact be fundamentally incompatible with the laws of thermodynamics~\cite{Guryanova2020a}. Thus, we here ask how much work may be extracted when only a single type of coarse-grained measurement is available to characterize the energy source. 

We derive the corresponding quantifiers of maximally extractable work under this operational constraint, and dub them Boltzmann and observational ergotropy, because they are implicitly defined by the average Boltzmann and observational entropy~\cite{Safranek2019,safranek2019quantum,strasberg2021first,Safranek2021,buscemi2022observational}, respectively. The first applies to a situation when the measurement outcomes are employed in the work extraction process, the second when the partially characterized source is no longer measured. Finally, we illustrate the effects of the operational constraints on work extraction from an evolving quantum state.

\medskip\noindent \textbf{Work extraction.} The amount of work that may be extracted from a quantum system is contingent on what type of system manipulation is experimentally possible~\cite{Niedenzu2019,kamin2021exergy,janovitch2022quantum,morrone2023daemonic,mula2023ergotropy,koshihara2023quantum}. Here, we will be concerned with a quantum system's ergotropy, which measures the amount of energy that can be extracted by a unitary transformation. It is defined as~\cite{Allahverdyan2004a}
\begin{equation}
W(\rho,\ham)=\tr[\ham\R]-\min_U\tr[\ham U\R U^\dag].
\label{eq:ergotropy}
\end{equation}
This can be written as a closed expression:
\begin{equation}\label{eq:singleergo}
W(\rho,\ham)= \tr[\ham(\R-\pi)],
\end{equation}
where, in the case of non-degenerate energy levels $\pi$ is the unique passive state for the tuple $(\R,\ham)$, meaning that it has the same eigenvalues as $\R$, multiplying energy eigenvectors in decreasing order (in the degenerate case, $\pi$ is a member of a family of passive states minimizing Eq.~\eqref{eq:ergotropy}). Interestingly, it is more efficient to extract the energy simultaneously from $N$ copies, in which case we obtain an asymptotic expression~\cite{alicki2013entanglement,hovhannisyan2013entanglement,campaioli2018quantum},
\begin{equation}\label{eq:multiergo}
W^\infty(\R):=\lim_{N\rightarrow \infty}\frac{W(\R^{\otimes N},\ham_N)}{N}=\tr[\ham(\R-\R_{\beta})].
\end{equation}
Here, $\R_\beta=e^{-\beta\ham}/Z$ is the thermal state with inverse temperature $\beta$ defined by the von Neumann entropy, $S_{\vN}(\R_{\beta})=S_{\vN}(\R)$, and $\ham_N:=\sum_{i=1}^N h_i$ with $h_i\cong H$ for all $i$ (That is, $H_N$ is a sum of isomorphic local terms; we notationally omit trivial terms on all other Hilbert spaces). $\rho_{\beta}$ always lower-bounds $\pi$ energetically: $\tr[H(\pi-\rho_{\beta})]\geq 0$. For the remainder of the paper we simplify the notation to $W(\rho)\equiv W(\rho,\ham)$. See Appendix A for a generalization of this result.

\begin{figure}[t]
\begin{center}
\includegraphics[width=1\hsize]{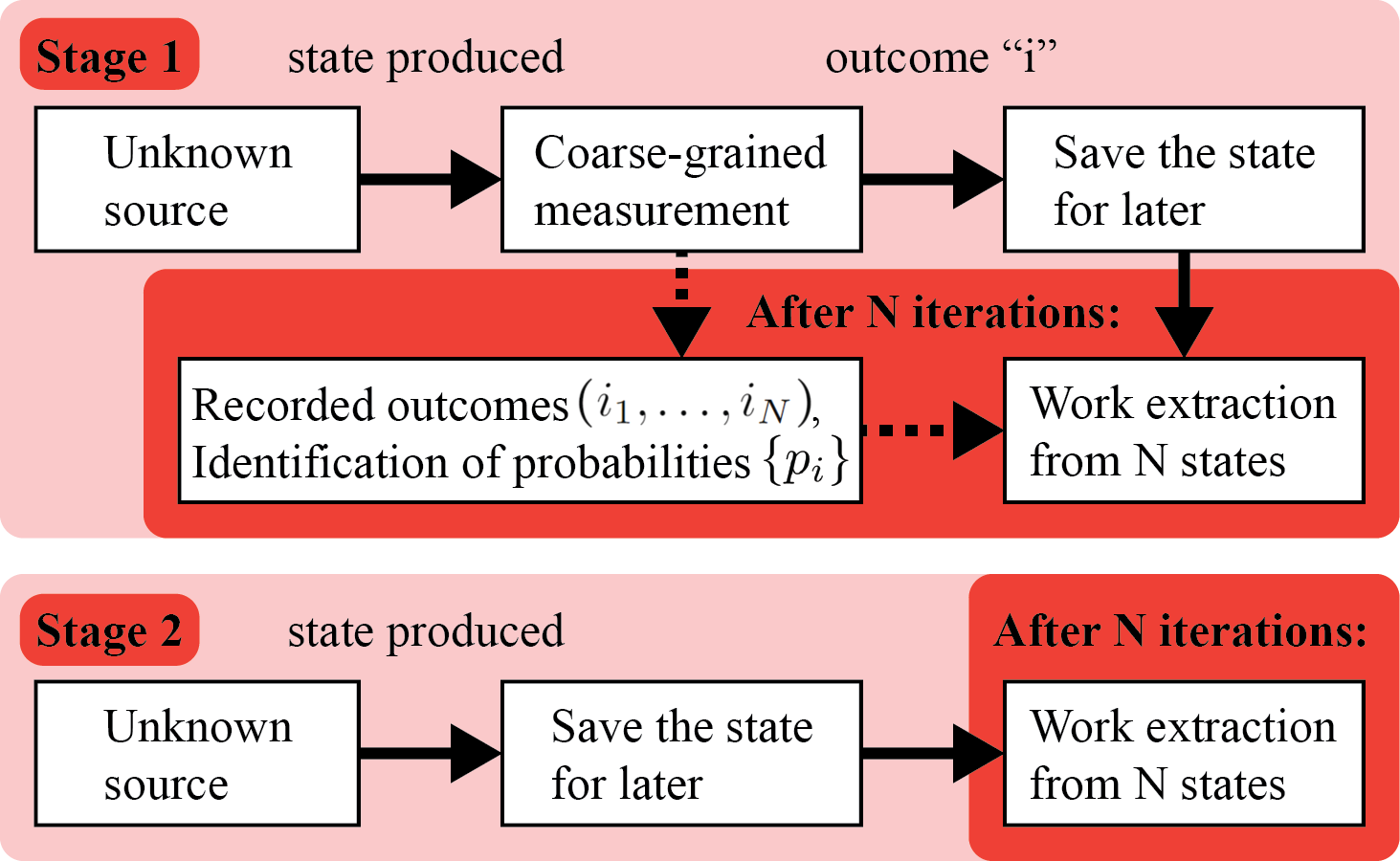}\\
\caption
{Scheme of work extraction from unknown sources.
In the limit of the large number of iterations $N$, the ergotropy per state is larger in Stage 1 at the cost of having to store the measurement outcomes. It is given by $W=\tr[\ham(\R-\R_{\beta})]$, where $\R_{\beta}$ is a thermal state with temperature $\beta$ implicitly defined by the mean Boltzmann entropy through $S_{\vN}(\R_{\beta})=S^B_\C$ in Stage 1, and by observational entropy through $S_{\vN}(\R_{\beta})=S_\C$ in Stage 2, respectively. See Eqs.~\eqref{eq:B_ergotropy_infty} and~\eqref{eq:o_ergotropy_infty}. The energy $\tr[\ham\R]$ is unknown, but it can be estimated from Eq.~\eqref{eq:bound_for_local_energy_cg}.}
\label{Fig:scheme}
\end{center}
\end{figure}

\medskip\noindent \textbf{Extraction scheme.} 
Consider a scenario where only a single type of measurement is available to characterize the source of quantum states, potentially comprising significantly fewer outcomes than the system's microscopic degrees of freedom. In particular, consider a set of orthogonal projectors $\C=\{\P_i\}$ with $\sum_i\P_i=\I$. The probability of outcome $i$ when measured on an unknown state $\R$ is $p_i=\tr[\P_i\R]$. After obtaining this outcome, the corresponding state is projected onto $\R_i=\P_i\R\P_i/p_i$. 
The measurement (=coarse-graining) also naturally defines the decomposition of the Hilbert space into subspaces-macrostates, $\HS=
\bigoplus_i\HS_i$, with each macrostate given by
\begin{equation}
\HS_i=\P_i\HS\P_i=\mathrm{span}_{\ket{\psi}}\{\P_i\ket{\psi}\}.
\end{equation}

We will study two work-extraction stages, depicted in Fig.~\ref{Fig:scheme}:

{\bf Stage 1} combines characterization the source with work extraction. We characterize the source by measuring a sufficiently large number $N$ of states. We record the outcomes and identify the corresponding probabilities $p_i$. Then, we extract work from these $N$ states, by a protocol that makes use of the records.

{\bf Stage 2} uses the already characterized source to extract energy from the quantum states it produces, without further measurements. The probabilities $p_i$ thus given, we extract energy from any number of states produced by the source.

\medskip\noindent \textbf{Partially random unitary extraction operations.} 
The original notion of ergotropy assumes perfect knowledge of the density matrix from which energy is extracted. This is necessary in order to find the best extraction unitary for that particular state. Here, in contrast, the state produced by the source is unknown; only the outcomes $i$, or probabilities $p_i$ are obtainable by measurement. Thus, given this incomplete information about the initial state, we are unable to find the unitary that extracts the energy perfectly. In order to address this, we design a protocol that makes use only of this incomplete information to extract energy. This protocol will be partially random, so also the extracted work will be random. However, we will be able to determine the average amount of extracted work from many copies of the same state and find the cases in which it is positive.  


The work extraction protocol consists of two unitary operations: First, a random unitary $\bigoplus_i \tilde U_i$ that randomizes states in each macrostate-subspace is applied. This is necessary to make the task tractable, by making the average state effectively known. 
Then a non-random global extraction unitary $U$, which makes use of the partial information, is applied to extract the remaining available energy. Thus, the total extraction operation, which is partially random, is
\begin{equation}\label{eq:cg_extraction_U}
U\bigoplus_i \tilde U_i.
\end{equation}
Unitary $U$ acts on the entire Hilbert space, and it is later optimized to take into account the knowledge of either outcomes $i$ or probabilities $p_i$ obtained from the measurements. $\tilde{U}_i$ are random unitaries, each acting on the corresponding macrostate-subspace $\HS_i$. We choose operations $\tilde{U}_i$ to be completely random, according to the Haar measure. This ensures that averaging over many 
realizations of the protocol leads to the following mathematical formula, defining a non-unitary operation
\begin{equation}\label{eq:global_extraction_operation}
    \mathcal{U}(\R)=U\Big( \int (\bigoplus_i\tilde{U}_i)\R (\bigoplus_i\tilde{U}_i^\dag) d\mu(\bigoplus_i\tilde{U}_i)\Big) U^\dag=U\R_\cg U^\dag.
\end{equation}
(See Supplemental Material for the proof and a simple analytical example.) Here, the coarse-grained state 
\begin{equation}
\R_\cg=\sum_i\frac{p_i}{V_i}\P_i
\end{equation}
is known, because both $p_i$ and $V_i$ are experimentally available, unlike the full original state $\R$. $V_i=\tr[\P_i]=\dim \HS_i$ is the volume of the macrostate (the number of its constituent distinct microstates), which depends solely on the measurement. The normalized ($\int d\mu(\bigoplus_i\tilde{U}_i)=1$) Haar measure factorizes into subspace unitary Haar measures as $d\mu(\bigoplus_i\tilde{U}_i)=d\mu(\tilde U_1)d\mu(\tilde U_2)\cdots$. 
We will use the formula~\eqref{eq:global_extraction_operation} when computing the average extracted work. 

Note that random unitaries are key to a number of theoretical protocols~\cite{Ohliger_2013,nahum2017quantum,Russell_2017,nahum2018operator,huang2020predicting,elben2020many,rossini2020measurement,rath2021quantum}, some of which were implemented in experiments~\cite{brydges2019probing,yu2021experimental}, and various methods to generate them have been developed~\cite{lundberg2004haar,mezzadri2007how}.


In the case of  simultaneous extraction from multiple copies, we choose the random unitaries to act on each individual copy, so that the global extraction operation amounts to
\begin{equation}\label{eq:global_extraction}
U(\bigoplus_i \tilde{U}_i)^{\otimes N}.
\end{equation}
$U$ is a global unitary acting on an $N$-partite state. After averaging, we obtain
\begin{equation}
\label{eq:global_extraction_operation_multiple}
\mathcal{U}(\R^{(1)}\otimes\cdots \otimes\R^{(N)})=U\R_\cg^{(1)}\otimes\cdots \otimes\R_\cg^{(N)}U^\dag.
\end{equation}

See the extraction protocol applied to Stages 1 and 2 in Fig.~\ref{fig:protocol}.

\medskip\noindent \textbf{Extracted work in Stage 1: with measurement.}
The total extracted work is obtained as the difference between the initial and the average final energy of the state. Its derivation is presented in Appendix B while here we present the main results.

In the case of extraction from a single copy of the initial state ($N=1$ in Figs.~\ref{Fig:scheme} and~\ref{fig:protocol}), the extracted work is measured by the \emph{Boltzmann ergotropy},
\begin{equation}\label{eq:B_ergotropy}
W_\C^{B}(\R)
=\mathrm{tr}\Big[\ham\Big(\R-\sum_ip_i\pi_i\Big)\Big],
\end{equation}
where $\pi_i$ is a passive state for the tuple $(\P_i/V_i,\ham)$. It describes the maximal amount of work extractable from a state produced by an unknown source, when measuring the state and using the outcome for the extraction protocol~\footnote{Note that Eq.~\eqref{eq:B_ergotropy} does not include the (Landauer) work that would be required to reset the measurement record~\cite{Landauer1961}.}. This maximal work is averaged over many realizations of the initial state. In particular, the global unitary $U$ in the extraction operation, Eq.~\eqref{eq:cg_extraction_U}, depends on the measurement outcome $i$ and thus is optimized for, and both the measurement outcomes and the unitaries $\tilde{U}_i$ are random --- these are averaged over.

In the case of simultaneous extraction from $N$ copies of the initial state, we derive the \emph{Boltzmann ergotropy in the large-$N$ limit},
\begin{equation}\label{eq:B_ergotropy_infty}
W_\C^{B\infty}(\R)=\tr[\ham(\R-\R_{\beta})].
\end{equation}
Temperature of the thermal state $\R_{\beta}=e^{-{\beta}\ham}/Z$ is implicitly defined by requiring that its von Neumann entropy equals the mean \emph{Boltzmann entropy} with coarse-graining $\C$,
\begin{equation}\label{eq:B_entropy_ergotropy}
\begin{split}
S_{\vN}(\R_{\beta})=S_\C^B.
\end{split}
\end{equation}
The mean Boltzmann entropy is defined as $S_\C^B=\sum_i p_i\ln V_i$, in which both $p_i$ and $V_i$ are experimentally accessible. 
Eq.~\eqref{eq:B_ergotropy_infty} defines the amount of extractable work \emph{per copy} in the large-$N$ limit when simultaneously extracting from $N$ copies of the initial state while using the measurement outcomes in the process. In particular, $U$ in Eq.~\eqref{eq:global_extraction} depends on the list of outcomes $(i_1,\dots, i_N)$ and random unitaries $\tilde{U}_i$ are averaged over. We have $W_\C^{B\infty}\geq W_\C^{B}$.

The extractable work depends largely on the amount of coarse-graining. For a fine-grained measurement projecting onto a pure state, we have $V_i=1$ and thus $S_\C^B=0$. This means that all the mean energy can be extracted. On the other hand, for a very coarse measurement in which $S_\C^B$ is large, very little work can be obtained. In fact, energy can even be lost, if $\tr[\ham\R_{\beta}]>\tr[\ham \R]$.

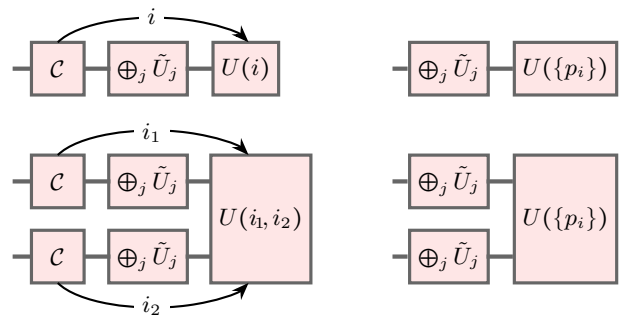
\begin{figure}
    \centering
    \begin{tikzpicture}
[roundnode/.style={circle, draw=red!60, fill=red!10, very thick, minimum size=7mm},
squarednode/.style={rectangle, draw=black!60, fill={rgb,255:red,254;green,230;blue,230}, very thick, minimum size=7mm},squarednode2/.style={rectangle, draw=black!60, fill={rgb,255:red,254;green,230;blue,230}, very thick, minimum height=17mm}, delta angle=30, >=Stealth
]
\node at ( 0,0) [squarednode] {$\C$};
\node at ( 1.2,0) [squarednode] {$\bigoplus_j \tilde U_j$};
\node at ( 2.5,0) [squarednode] {$U(i)$};
\draw[-][black!60, ultra thick] (-0.6,0) -- (-0.35,0);
\draw[-][black!60, ultra thick] (0.35,0) -- (0.7,0);
\draw[-][black!60, ultra thick] (1.7,0) -- (2.05,0);
arc [start angle=30, delta angle=30] -- cycle;
\draw[thick,->] (0,.35) arc [start angle=160, end angle=
20, x radius=13.5mm, y radius=5mm];
\newcommand{\ydist}{-1.5}
\node at ( 0,\ydist) [squarednode] {$\C$};
\node at ( 1.2,\ydist) [squarednode] {$\bigoplus_j \tilde U_j$};
\node at ( 0,\ydist-1) [squarednode] {$\C$};
\node at ( 1.2,\ydist-1) [squarednode] {$\bigoplus_j \tilde U_j$};
\node at ( 2.7,\ydist-.5) [squarednode2] {$U(i_{\!1}\!,i_{2})$
};
\draw[-][black!60, ultra thick] (-0.6,\ydist) -- (-0.35,\ydist);
\draw[-][black!60, ultra thick] (0.35,\ydist) -- (0.7,\ydist);
\draw[-][black!60, ultra thick] (1.7,\ydist) -- (2.05,\ydist);
\draw[-][black!60, ultra thick] (-0.6,\ydist-1) -- (-0.35,\ydist-1);
\draw[-][black!60, ultra thick] (0.35,\ydist-1) -- (0.7,\ydist-1);
\draw[-][black!60, ultra thick] (1.7,\ydist-1) -- (2.05,\ydist-1);
arc [start angle=30, delta angle=30] -- cycle;
\draw[thick,->] (0,\ydist+.35) arc [start angle=160, end angle=
20, x radius=13.5mm, y radius=5mm];
\draw[thick,->] (0,\ydist-1.35) arc [start angle=160, end angle=
20, x radius=13.5mm, y radius=-5mm];
\newcommand{\xdist}{4}
\node at ( 1.2+\xdist,0) [squarednode] {$\bigoplus_j \tilde U_j$};
\node at ( 2.75+\xdist,0) [squarednode] {$U(\{p_i\})$};
\draw[-][black!60, ultra thick] (0.45+\xdist,0) -- (0.65+\xdist,0);
\draw[-][black!60, ultra thick] (1.7+\xdist,0) -- (2.05+\xdist,0);
\node at ( 1.2+\xdist,\ydist) [squarednode] {$\bigoplus_j \tilde U_j$};
\node at ( 1.2+\xdist,\ydist-1) [squarednode] {$\bigoplus_j \tilde U_j$};
\node at ( 2.75+\xdist,\ydist-.5) [squarednode2] {$U(\{p_i\})$
};
\draw[-][black!60, ultra thick] (0.45+\xdist,\ydist) -- (0.65+\xdist,\ydist);
\draw[-][black!60, ultra thick] (1.7+\xdist,\ydist) -- (2.05+\xdist,\ydist);
\draw[-][black!60, ultra thick] (0.45+\xdist,\ydist-1) -- (0.65+\xdist,\ydist-1);
\draw[-][black!60, ultra thick] (1.7+\xdist,\ydist-1) -- (2.05+\xdist,\ydist-1);
arc [start angle=30, delta angle=30] -- cycle;
\node[fill=white] at (1.25,0.7) {$i$};
\node[fill=white] at (1.25,-.85) {$i_1$};
\node[fill=white] at (1.25,-3.15) {$i_2$};
\end{tikzpicture}
    \caption{The work extraction protocol in Stage 1 (left) and Stage 2 (right), for $N=1$ (single state; top) and $N=2$ (two states; bottom). $\C$ denotes the measurement.}
    \label{fig:protocol}
\end{figure}

\medskip\noindent \textbf{Extracted work in Stage 2: no measurement.
} 
Assuming that the source has been characterized, with probabilities $p_i$ known, how much energy can be extracted without further measurement? The derivation for the following results can be found in Appendix C.

In analogy to Eq.~\eqref{eq:B_ergotropy}, the extracted work from a single state is measured by the \emph{observational ergotropy},
\begin{equation}\label{eq:O_ergotropy_i}
W_\C(\R)=\tr[\ham(\R-\pi_\cg)],
\end{equation}
where $\pi_\cg$ is a passive state for the tuple $(\R_\cg,\ham)$. 
It describes the maximally-extractable work from a state produced by an unknown source characterized by a set of probabilities $\{p_i\}$, without further measurements.
In particular, the extraction operation, Eq.~\eqref{eq:cg_extraction_U}, is applied directly on the state which is not measured beforehand. The global unitary $U$ depends on the probabilities $\{p_i\}$ and the random unitaries $\tilde{U}_i$ are averaged over.

In the case of simultaneous extraction from $N$ copies of the initial state, we obtain the \emph{observational ergotropy in the large-$N$ limit},
\begin{equation}\label{eq:o_ergotropy_infty}
W_\C^{\infty}(\R)=\tr[\ham(\R-\R_{\beta'})].
\end{equation}
Temperature $\beta'$ of the thermal state is implicitly defined by requiring that its von Neumann entropy equals the observational entropy,  
\begin{equation}\label{eq:O_ergotropy}
\begin{split}
S_{\vN}(\R_{\beta'})=S_\C.
\end{split}
\end{equation}
Observational entropy~\cite{Safranek2021} is the sum of Shannon and mean Boltzmann entropy, $S_\C=S_\C^{S\!\!\;h}+S_\C^B=-\sum_i p_i \ln p_i+\sum_i p_i\ln V_i$. Eq.~\eqref{eq:o_ergotropy_infty} measures the average amount of extractable work per copy when extracting simultaneously from a large number of copies of the state, produced by the source characterized solely by the probabilities $\{p_i\}$. The global unitary $U$ in Eq.~\eqref{eq:global_extraction} depends on the  probabilities and $\tilde U_i$ are averaged over. We have $W_\C^{\infty}\geq W_\C$.

Due to the conditional extraction, Stage~1 leads to a larger extractable work than Stage~2, $W_\C^{B}\geq W_\C$ and $W_\C^{B\infty}\geq W_\C^{\infty}$. The difference in Eqs.~\eqref{eq:B_entropy_ergotropy} and~\eqref{eq:O_ergotropy} is given by the (Shannon) entropy of measurement. An example of observational ergotropy and the corresponding entropy is depicted in Fig.~\ref{Fig:ergotropy}.

\medskip\noindent \textbf{Bounds on observational ergotropy.} 
Computing the Boltzmann or observational ergotropy requires knowledge of the mean initial energy, which is unknown but possible to estimate from the partial knowledge given by distribution $\{p_i\}$.

Consider local energy coarse-grainings,
\begin{equation}\label{eq:local_energy_cg}
\C=\{\P_{E_1}\otimes \P_{E_2}\},
\end{equation}
studied in observational entropy literature~\cite{safranek2019quantum,Safranek2021}.  $\P_{E_{1}}=\sum_{E_1'\in [E_1,E_1+\Delta E)}\pro{E_1'}{E_1'}$ and $\P_{E_{2}}$ (analogous) are coarse-grained projectors on local energies with resolution $\Delta E$. The mean energy can be estimated from
\begin{equation}\label{eq:bound_for_local_energy_cg}
\abs{\tr[\ham \R]-\tr[\ham \R_\cg]}\leq 2\norm{H_{\mathrm{int}}}+2\Delta E,
\end{equation}
where the full Hamiltonian is $\ham=\ham_1+\ham_2+\ham_{\mathrm{int}}$, and $\norm{~}$ denotes the operator norm. 
When increasing the number of partitions to $k$, we obtain $\abs{\tr[\ham \R]-\tr[\ham \R_\cg]}\leq 2\norm{\ham_{\mathrm{int}}}+k\Delta E$. The first term $\norm{\ham_{\mathrm{int}}}$ also scales linearly with $k$, representing a finite-size effect. See Supplemental Material for details.

Energy can also be estimated in the case of completely general coarse-graining~\cite{safranek2023expectation}.

\begin{figure}[t]
\begin{center}
\includegraphics[width=1\hsize]{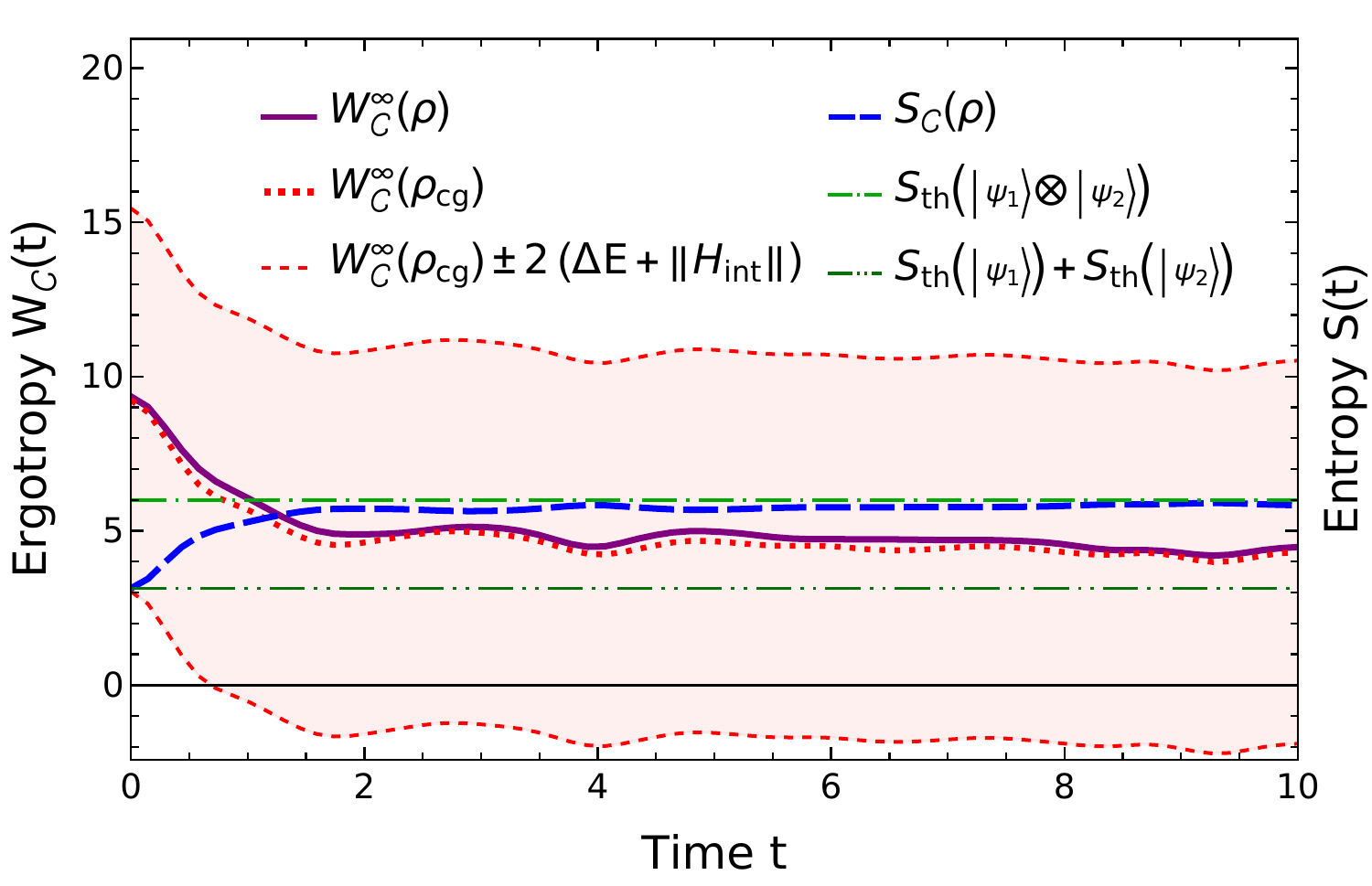}\\
\caption
{Observational ergotropy and observational entropy as a function of time, for a thermalizing system of four particles. Here, $\C$ is chosen to be a local energy coarse-graining (Eq.~\eqref{eq:local_energy_cg}); the initial state $\ket{\psi_1}\otimes \ket{\psi_2}$, $\ket{\psi_1}=\ket{000000}$ and $\ket{\psi_2}=\ket{111100}$, evolves with the Hamiltonian given by Eq.~\eqref{FermionHam}, with $T=V=1$ and $T'=V'=0.96$. We choose the energy resolution $\Delta E=(E_1-E_0)/2$, where $E_0$ and $E_1$ are the ground and first excited state energy, respectively. True ergotropy $W_\C^\infty(\R)$ (solid purple) is unknown to the experimenter. However, they can estimate the value of $W_\C^\infty(\R_{\cg})$ (red dotted), $\R_\cg=\sum_{E_1,E_2}\frac{p_{E_1E_2}(t)}{V_{E_1}V_{E_2}}\P_{E_1}\otimes\P_{E_2}$, and be sure that the true value lies within the light-red shaded region, between the red-dashed lines representing $W_\C^\infty(\R_{\cg})\pm 2(\norm{H_{\mathrm{int}}}+\Delta E)$. We compare this with observational entropy $S_\C(\R)$ (blue long-dashed) called non-equilibrium thermodynamic entropy for this particular coarse-graining~\cite{Safranek2021}. It is bounded by the sum of initial thermodynamic entropies (dark green dot-dot-dashed) and by the final thermodynamic entropy (green dot-dashed), $S_{\th}\equiv S_{C_E}$, defined as observational entropy with global energy coarse-graining with the same resolution $\Delta E$~\cite{safranek2019quantum,Safranek2021}.
Observational ergotropy is directly defined by observational entropy, and inversely related to it, as per Eqs.~\eqref{eq:o_ergotropy_infty} and~\eqref{eq:O_ergotropy}.
}
\label{Fig:ergotropy}
\end{center}
\end{figure}

\medskip\noindent \textbf{Example: ergotropy of local energy coarse-graining.} 
We illustrate the essence of our result by examining observational ergotropy as a function of time for a thermalizing system. We assume that we can measure local energies, described by coarse-graining~\eqref{eq:local_energy_cg}. We consider a widely-used and generic  one-dimensional fermionic Hamiltonian~\cite{santos2010onset}, with the nearest neighbor and next-nearest neighbor interaction, describing interacting particles hopping between $k$'th and $l$'th site as
\begin{equation}\label{FermionHam}
\ham^{(k:l)} \!=\! \sum_{i=k}^{l} \!- T f_i^{\dagger} f_{i+1} - T' f_i^{\dagger}f_{i+2}+ h.c.  +V n_i   n_{i+1}
 +V' n_{i}  n_{i+2} .
\end{equation}
(Terms with $f_{l+1}$, $f_{l+2}$, $f_{k-1}$, and $f_{k-2}$  are not included in the sum.) $f_i$ and $f_i^{\dagger}$ are the fermionic annihilation and creation operators for site $i$.  $n_i= f_i^{\dagger} f_i$ is the local density operator. We take the full Hamiltonian $\ham\equiv \ham^{(1:L)}$, where $L$ is the length of the chain, with a division between two equal sized subsystems $\ham_1\equiv \ham^{(1:L/2)}$, $\ham_2\equiv \ham^{(L/2+1:L)}$, and $H_{\mathrm{int}}=H-H_1-H_2$. We also employ hard wall boundary conditions.

In Fig.~\ref{Fig:ergotropy} we plot observational ergotropy $W_\C^\infty(\ket{\psi_t})$ for extracting energy from state $\ket{\psi_t}=\exp(-i\ham t)\ket{\psi_0}$, together with its estimates. We compare this with observational entropy to which it is reciprocally related through Eq.~\eqref{eq:o_ergotropy_infty}. 
There is a point where the lower bound on ergotropy crosses zero, in which the experimenter, given the available information, will stop characterizing the state as useful, i.e., they cannot be certain that it provides energy. As the state evolves, the system thermalizes, and the opportunity to extract work diminishes.

As the full Hamiltonian preserves the total number of particles $n$, the relevant Hilbert space explored by the system during time evolution is $\HS=\bigoplus_{k=0}^n \HS_{k}\otimes \HS_{n-k}$. This is important for correctly computing the accessible macrostate volumes $V_{E_1E_2}$. To describe situations with an unknown particle number, one has to employ additional coarse-graining in local particle numbers~\cite{Safranek2021}. 

\textbf{Discussion and Conclusions.} 
To relate ergotropy to realistic scenarios of work extraction, previous works have considered constraints such as a restriction to local~\cite{alicki2013entanglement,Hovhannisyan2013,Binder2015a,campaioli2017enhancing,alimuddin2019bound,puliyil2022thermodynamic}, Gaussian~\cite{Brown2016}, or incoherent operations~\cite{Francica2020}. Here, we have addressed the remaining open problem which is the assumption of perfect knowledge of the initial state.

Assuming that a  source of unknown states can be characterized only by a single type of coarse-grained measurement, we designed the extraction protocol as follows. A random unitary is applied on each measurement subspace. 
Then we apply a global unitary operation that optimizes the energy extraction by taking the knowledge obtained from the measurement into account.
Because of the randomness, also the work extracted is random. However, because the unitaries were picked with the Haar measure, the extracted work average is computable. 
This allows to determine the work output of the source.

The protocol results in two notions of ergotropy: Boltzmann ergotropy, which measures the extracted energy when the measurement result is conditionally taken into account (Eq.~\ref{eq:B_ergotropy_infty}), and observational ergotropy for unconditional extraction (Eq.~\ref{eq:o_ergotropy_infty}). The energy difference between the two cases results from the difference in entropy of corresponding thermal states (Eq.~\ref{eq:B_entropy_ergotropy} and Eq.~\ref{eq:O_ergotropy}). The two deviate exactly by the (Shannon) entropy of measurement $S^{S\!\!\;h}=-\sum_ip_i\ln p_i$ which lower-bounds the work required for erasing the measurement record~\cite{Sagawa2009,Reeb2014,Parrondo2015}: $W_\mathrm{erasure}\geq \beta^{-1}S^{S\!\!\;h}$ with $\beta$ the inverse temperature of the heat bath used during erasure. Note here that we have \textit{not} attempted to quantify the energy associated with the measurement itself (which indeed diverges for perfect projective measurement~\cite{Guryanova2020a}).

This work applies in two cases: first, to high-dimensional external sources, i.e., not prepared by an experimenter, on which full quantum tomography is not viable. The second case is that of imperfectly controlled systems, such as quantum batteries. Given perfect control over the charging procedure, one knows the state of the charged battery exactly. Therefore ergotropy is the relevant figure of merit \cite{andolina2019extractable, barra2019dissipative, hovhannisyan2020charging, delmonte2021characterization}. However, a certain lack of control is inevitable -- e.g. in the form of unknown disorder~\cite{rossini2019many-body, ghosh2020enhancement, rossini2020quantum, rosa2020ultra-stable,caravelli2020random, zhao2021quantum, kim2022operator, arjmandi2022enhancing}, uncertain time of charging~\cite{mitchison2021charging, yao2022optimal}, or batteries charged from an initially unknown state~\cite{Landi2021}. In all these cases, the final battery state is not fully determined. Observational ergotropy then gives an experimentally verifiable lower bound on the amount of energy that can be extracted. As such, it is a realistic figure of merit for characterizing quantum batteries.

One can also find applications from the theoretical perspective. These appear whenever there is a limit on which measurement the experimenter can perform, in the spirit of E.T.~Jaynes' statement regarding the Gibbs (mixing) paradox~\cite{jaynes1992gibbs}: ``The amount of useful work that we can extract from any system depends - obviously and necessarily - on how much ``subjective'' information we have about its microstate because that tells us which interactions will extract energy and which will not.'' Here as well, depending on the measurement, the outcome will determine the best extraction unitary and in turn the amount of extractable work.

\medskip\noindent \textbf{Acknowledgments.} DR and D\v{S} acknowledge the support from the Institute for Basic Science in Korea (IBS-R024-D1). FCB acknowledges support from grant number FQXi-RFP-IPW-1910 from the Foundational Questions Institute and Fetzer Franklin Fund, a donor-advised fund of the Silicon Valley Community Foundation. D\v{S} thanks Anthony Aguirre, Joshua M. Deutsch, Joseph Schindler, and Susanne Still for discussing theory years in the making. We acknowledge Philipp Strasberg for the excellent feedback on the first version of this manuscript and for discussing with us the examples mentioned in the Supplemental Material.


\medskip\noindent\textbf{Appendix A: simultaneous extraction generalized to a larger class of states.}
The essence of Eq.~\eqref{eq:multiergo} may equally be expressed as
\begin{equation}
\label{eq:largeN}
\lim_{N\rightarrow \infty}\min_U\frac{\tr[\ham_N U\R^{\otimes N}U^\dag]}{N}=\tr[\ham \R_{\beta}],
\end{equation}
where $S_{\vN}(\R_{\beta})=S_{\vN}(\R)$.

In the Supplemental Material, we derive its generalization,
\begin{equation}
\label{eq:largeN2}
\lim_{N\rightarrow \infty}\min_U\frac{\tr[\ham_N U\bigotimes_{i}\R_i^{\otimes p_iN}U^\dag]}{N}=\tr[\ham \R_{\beta}],
\end{equation}
where $S_{\vN}(\R_{\beta})=\sum_i p_i S_{\vN}(\R_i)$, assuming that the limit exists. Here and in the following, $\R_i^{\otimes p_iN}$ is shorthand for $\R_i^{\otimes \lfloor p_iN\rfloor}$ and $\rho_i^{\otimes 0}=1$ inside the limit. See Supplemental Material for details. $\{p_i\}$ is a set of probabilities, $\sum_ip_i=1$, and $\{\R_i\}$ is any set of density matrices.  The original formula is recovered for $p_1=1$.

\medskip\noindent \textbf{Appendix B: derivation of extracted work in Stage 1: with measurement.} See the extraction protocol in Fig.~\ref{fig:protocol} (left). 
The total extracted work is given by the difference between the initial and the final energy of the state, which for measurement outcome $i$ is
\begin{equation}\label{eq:energetics}
W_{\mathrm{single}}^{i}=\tr[\ham\R]-\tr[\ham U\big(\bigoplus_j\tilde{U}_j\big)\R_i\big(\bigoplus_j\tilde{U}_j^\dag\big)U^\dag],
\end{equation}
where $\R_i=\P_i\R\P_i/p_i$. 
In cases where the outcome is $i$, the experimenter on average extracts
\begin{equation}\label{eq:wi_on_average}
W^i=\int W_{\mathrm{single}}^{i}\, d\mu(\bigoplus_j\tilde{U}_j)=\tr[\ham\R]-\tr[\ham U(\P_i/V_i)U^\dag].
\end{equation}
Here we used $\mathcal{U}(\R_i)=U\P_i/V_iU^\dag$.
Maximizing over the global unitary and using Eq.~\eqref{eq:singleergo}, we define the \emph{Boltzmann ergotropy corresponding to outcome $i$},
\begin{equation}\label{eq:Bi_ergotropy}
W_\C^{B,i}(\R)=\max_{U}\tr[\ham(\R-U(\P_i/V_i)U^\dag)]=\tr[\ham(\R-\pi_i)].
\end{equation}
$\pi_i$ is a passive state for the tuple $(\P_i/V_i,\ham)$.
Averaging over all possible outcomes defines the \emph{Boltzmann ergotropy},
\begin{equation}\label{eq:B_ergotropy2}
W_\C^{B}(\R)
=\mathrm{tr}\Big[\ham\Big(\R-\sum_ip_i\pi_i\Big)\Big].
\end{equation}

In the case of simultaneous extraction in the limit of large $N$, the state after a series of measurements is (\emph{almost surely}) described by
\begin{equation}\label{eq:after_measurement_state}
\R_N=\bigotimes_i\R_i^{\otimes p_i N}
\end{equation}
This is up to reordering of the outcomes: while the order affects the specific extraction unitary, all such permutations lead to the same ergotropy (since these states are energetically equivalent). This is derived by using the law of large numbers; see Supplemental Material, which includes related Ref.~\cite{cover1999elements}.

Using the same logic as in Eq.~\eqref{eq:wi_on_average} when applying Eq.~\eqref{eq:global_extraction_operation_multiple} on Eq.~\eqref{eq:after_measurement_state}, and then using Eq.~\eqref{eq:largeN2} we derive the Boltzmann ergotropy in the large-$N$ limit,
\begin{equation}\label{eq:B_ergotropy_infty2}
W_\C^{B\infty}(\R):=\lim_{N\rightarrow \infty}\frac{W_\C^B(\R^{\otimes N})}{N}=\tr[\ham(\R-\R_{\beta})].
\end{equation}
Temperature of the thermal state $\R_{\beta}=e^{-{\beta}\ham}/Z$ is implicitly defined by requiring that its von Neumann entropy equals the mean \emph{Boltzmann entropy} with coarse-graining~$\C$,
\begin{equation}\label{eq:B_entropy_ergotropy2}
\begin{split}
S_{\vN}(\R_{\beta})=\sum_i p_i \ln V_i=:S_\C^B.
\end{split}
\end{equation}
See Supplemental Material for details.

\medskip\noindent \textbf{Appendix C: derivation of extracted work in Stage 2: no measurement.
} See the extraction protocol in Fig.~\ref{fig:protocol} (right). As in Eq.~\eqref{eq:energetics}, we derive this energy as the energy difference between the initial and the final state,
\begin{equation}
W^\mathrm{no-meas.}_{\mathrm{single}}=\tr[\ham\R]-\tr[\ham U\big(\bigoplus_i\tilde{U}_i\big)\R\big(\bigoplus_i\tilde{U}_i^\dag\big)U^\dag].
\end{equation}
Averaging over the random unitaries using Eq.~\eqref{eq:global_extraction_operation}, and maximizing over the extraction unitary defines \emph{observational ergotropy},
\begin{equation}\label{eq:O_ergotropy_i2}
W_\C(\R)=\tr[\ham\R]-\min_U\tr[\ham U\R_\cg U^\dag]=\tr[\ham(\R-\pi_\cg)].
\end{equation}
Here, $\pi_\cg$ is a passive state for the tuple $(\R_\cg,\ham)$.

To understand the interpretation of Eq.~\eqref{eq:O_ergotropy_i2} in detail, we show that we obtain the same formula when doing the operations in the reverse order: first applying the extraction operation, and then averaging over the resulting work. Consider the optimal $U$ (which we denote $U_{\mathrm{opt}}$) that transforms state $\R_\cg$ into the passive state $\pi_\cg$. The extracted work \emph{in a single realization} is random and given by the difference between the initial and the final state energy,
\[
W^\mathrm{no-meas.}_{\mathrm{single,opt.}}=\tr[\ham\R]-\tr[\ham U_{\mathrm{opt}}\big(\bigoplus_i\tilde{U}_i\big)\R\big(\bigoplus_i\tilde{U}_i^\dag\big)U_{\mathrm{opt}}^\dag].
\]
Thus, on average, the experimenter extracts $\mean{W}=\int W^\mathrm{no-meas.}_{\mathrm{single,opt.}}\, d\mu(\bigoplus_j\tilde{U}_j)$. Due to the linearity of the trace and operator multiplication, this equals $W_\C(\R)$.

In the case of simultaneous extraction, we have
\begin{equation}
\begin{split}
W_\C(\R^{\otimes N})&=\tr[\ham_N\R^{\otimes N}]-\min_U\tr[\ham_N\mathcal{U}\big(\R^{\otimes N}\big)]\\
&=N\tr[\ham\R]-\min_U\tr[\ham_N U\R_\cg^{\otimes N}U^\dag].
\end{split}
\end{equation}
From this, we obtain observational ergotropy in the large-$N$ limit,
\begin{equation}\label{eq:o_ergotropy_infty2}
W_\C^{\infty}(\R):=\lim_{N\rightarrow \infty}\frac{W_\C(\R^{\otimes N})}{N}=\tr[\ham(\R-\R_{\beta'})].
\end{equation}
As per Eq.~\eqref{eq:largeN}, temperature $\beta'$ of the thermal state is implicitly defined by requiring that its von Neumann entropy equals the observational entropy $S_\C$~\cite{Safranek2021}, 
\begin{equation}\label{eq:O_ergotropy2}
\begin{split}
S_{\vN}(\R_{\beta'})&=S_{\vN}(\R_\cg)=-\sum_i p_i \ln p_i+\sum_i p_i\ln V_i\\
&=S_\C^{S\!\!\;h}+S_\C^B=S_\C,
\end{split}
\end{equation}
which is the sum of Shannon and mean Boltzmann entropy.


\bibliography{main.bib}



\setcounter{section}{0}
\setcounter{theorem}{0}
\setcounter{corollary}{0}

\section*{Supplemental Material
}

In this Supplemental Material, we provide several proofs and derivations, as well as examples. It contains: Sec.~\ref{sec:1}: generalized formula for the maximal extracted work. Sec.~\ref{app:averaging}: proof that averaging of any initial state over local unitaries leads to the coarse-grained state. Sec.~\ref{sec:2}: limit state after $N$ measurements in stage 1, Sec.~\ref{sec:3}: Boltzmann ergotropy in the large $N$ limit,  Sec.~\ref{sec:4}: examples on a three-level system, with a special focus on understanding the coarse-grained unitary, Sec.~\ref{sec:5}: bound on energy for local energy coarse-graining.

\section{Generalization of the maximally extracted work}\label{sec:1}

Here, we will derive a generalization of the formula for the maximally extracted work,
\begin{equation}\label{eq:original_sequnce_proof_app}
\lim_{N\rightarrow \infty}\min_U\frac{\tr[\ham_N U\bigotimes_{i}\R_i^{\otimes p_iN}U^\dag]}{N}=\tr[\ham \R_{\beta}],
\end{equation}
where $S_{\vN}(\R_{\beta})=\sum_i p_i S_{\vN}(\R_i)$, assuming that the limit exists. $\{p_i\}$ is a set of probabilities, $\sum_ip_i=1$, and $\{\R_i\}$ is any set of density matrices. Eq.~\eqref{eq:original_sequnce_proof_app} corresponds to Eq.~(20) 
in the main text.

As $p_iN$ is not necessarily an integer (which is irrelevant for large $N$), to be mathematically exact, we define
\begin{equation}\label{eq:limit_state_floor}
\R^N\equiv\bigotimes_{i}\R_i^{\otimes p_iN}:=\R_1^{\otimes(N-\sum_{i\geq 2}\lfloor p_i N \rfloor)}\bigotimes_{i\geq 2}\R_i^{\otimes \lfloor p_i N \rfloor}.
\end{equation}
This ensures that there are exactly $N$ states and each state is exponentiated to an integer.

The Hamiltonian is defined as
\begin{equation}
\ham_N:=\sum_{i=1}^N h_i=\ham\!\otimes I\!\otimes\!\cdots\!\otimes\! I\,+\,I\!\otimes\!\ham\!\otimes I\!\otimes\!\cdots\!\otimes\! I\,+\,\cdots
\end{equation}
with $h_i\cong H$ for all $i$ (that is, $H_N$ is a sum of isomorphic local terms).

We have to assume that the limit in Eq.~\eqref{eq:original_sequnce_proof_app} exists because, while it seems quite natural, it is difficult to prove. Unlike in the original statement, the sequence~\eqref{eq:original_sequnce_proof_app} is not monotonously decreasing. This is because of implementation of $\R^N$ using integers, Eq.~\eqref{eq:limit_state_floor}. Elements of the sequence may temporarily spike due to the discrete nature $\R^N$, even though the sequence is expected to go down most of the time. It is clear though that the right hand side of Eq.~\eqref{eq:original_sequnce_proof_app} is a lower bound, which follows from the variational principle of statistical mechanics, which asserts that the Gibbs canonical density matrix  minimizes the free energy, 
\begin{equation}
\tr[\R\ham]-\beta^{-1}S_{\vN}(\R)\geq 
\tr[\R_\beta\ham]-\beta^{-1}S_{\vN}(\R_\beta).
\end{equation}

To prove Eq.~\eqref{eq:original_sequnce_proof_app}, we are going to show that for any $\epsilon>0$ we find an element of the sequence closer to $\tr[\ham \R_{\beta}]$ than $\epsilon$. Because we assume that the limit exists, it must be equal to this number. We assume that $\R^N$ is given by Eq.~\eqref{eq:limit_state_floor}, although the exact way how the $\R^N$ is represented using floor functions in order to be properly defined will not matter for the argument.

We pick a subsequence of the sequence obtained by substitution $N\rightarrow M\!N$ (product of $M$ and $N$). This gives
\begin{equation}
\R^{M\!N}=\R_1^{\otimes({M\!N}-\sum_{i\geq 2}\lfloor p_i {M\!N} \rfloor)}\bigotimes_{i\geq 2}\R_i^{\otimes \lfloor p_i {M\!N} \rfloor}.
\end{equation}
For now, we assume that $M$ is some large integer. Consider a similar state
\begin{equation}
\big(\R^{M}\big)^{\otimes N}=\Big(\R_1^{\otimes({M}-\sum_{i\geq 2}\lfloor p_i {M} \rfloor)}\bigotimes_{i\geq 2}\R_i^{\otimes \floor{ p_i {M}}}\Big)^{\otimes N}.
\end{equation}
Energies of $\R^{M\!N}$ and $\big(\R^{M}\big)^{\otimes N}$ are very similar. In fact, we have
\begin{equation}\label{eq:differenceMNstates}
\tr[\ham_{M\!N}\R^{M\!N}]=\tr[\ham_{M\!N}\big(\R^{M}\big)^{\otimes N}]+k\big(NM g_M+N f_N\big).
\end{equation}
where $\lim_{N\rightarrow \infty} f_N=\lim_{M\rightarrow \infty} g_M=0$ and $k$ is the number of different $i's$ (number of elements in sets $\{p_i\}$ or $\{\R_i\}$; this is by construction independent of $N$ and $M$).
This is because the exponents associated with a given $i$ differ only slightly for large enough $N$ and $M$ (Recall that permutation of density matrices does not change the energy.) For example, for an $i\geq 2$ we have
\begin{equation}
\begin{split}
\floor{p_i M N}-\floor{ p_i M }N&=N\left(\frac{\floor{p_iM N}}{N}-\floor{ p_i M }\right)\\
&=N\left(p_iM-\floor{ p_i M }+f_N\right)\\
&=NM\left(p_i-\frac{\floor{ p_i M }}{M}\right)+Nf_N\\
&=NM g_M+N f_N
\end{split}
\end{equation}
where $\lim_{N\rightarrow \infty} f_N=\lim_{M\rightarrow \infty} g_M=0$, because $\lim_{N\rightarrow \infty}  \frac{\floor{x N}}{N}=x$. Since there are $k$ indices, the total difference is of order $k$ times the above.

According to $\lim_{N\rightarrow \infty}\min_U\tr[\ham_N U\R^{\otimes N}U^\dag]/N=\tr[\ham \R_{\beta}]$, where $S_{\vN}(\R_{\beta})=S_{\vN}(\R)$, (Eq.~(19) 
in the main text), we have
\begin{equation}
\lim_{N\rightarrow \infty}\!\min_U\frac{\tr[\!\ham_{M\!N} U\big(\R^{M}\big)^{\!\otimes N}U^\dag\!]}{MN}\!=\!\frac{\tr[\ham_M \R_{\beta'}^{\otimes M}]}{M}\!=\!\tr[\ham\!\R_{\beta'}].
\end{equation}
where $S_{\vN}(\R_{\beta'}^{\otimes M})=S_{\vN}(\R^M)$, which written explicitly gives
\begin{equation}\label{eq:rhobetaM}
\begin{split}
S_{\vN}(\R_{\beta'})&=\bigg(1-\sum_{i\geq 2}\frac{\floor{p_i {M}}}{M}\bigg) S(\R_1) +\sum_{i\geq 2} \frac{\floor{ p_i {M}}}{M} S(\R_i).
\end{split}
\end{equation}
Eq.~\eqref{eq:differenceMNstates} also holds for the transformed states, since the same argument can be made. Thus we have,
\begin{equation}
\begin{split}
&\lim_{N\rightarrow \infty}\!\min_U\frac{\tr[\!\ham_{M\!N} U\R^{M\!N}U^\dag\!]}{MN}\\
&=\lim_{N\rightarrow \infty}\!\min_U\frac{\tr[\!\ham_{M\!N} U\big(\R^{M}\big)^{\!\otimes N}U^\dag\!]}{MN}+k g_M+k\frac{f_N}{M}\\
&=\tr[\ham\!\R_{\beta'}]+k g_M,
\end{split}
\end{equation}
where $\R_{\beta'}$ is given by Eq.~\eqref{eq:rhobetaM}. Taking the limit $M\rightarrow \infty$, the thermal state energy $\tr[\ham\!\R_{\beta'}]$ converges to $\tr[\ham\!\R_{\beta}]$ and $g_M$ converges to zero. Therefore, for any $\epsilon>0$ we can find $M$ and $N$ such that an element of sequence Eq.~\eqref{eq:original_sequnce_proof_app}, defined by $\R^{M\!N}$, is closer to $\tr[\ham\!\R_{\beta}]$ than $\epsilon$. Because we assume that the sequence converges, it must converge to this value.

\section{Averaging of a state over the direct sum of unitaries with the Haar measure leads to the coarse-grained state---proof}\label{app:averaging}
Here we prove the statement
\begin{equation}\label{eq:global_extraction_operation_app}
    \mathcal{U}(\R)=U\Big( \int (\bigoplus_i\tilde{U}_i)\R (\bigoplus_i\tilde{U}_i^\dag) d\mu(\bigoplus_i\tilde{U}_i)\Big) U^\dag=U\R_\cg U^\dag,
\end{equation}
where 
\[
\R_\cg=\sum_i\frac{p_i}{V_i}\P_i,
\]
and $d\mu$ is the normalized ($\int d\mu(\bigoplus_i\tilde{U}_i)=1$) Haar measure, which is Eq. (7) in the main text. The global unitary $U$ is superfluous, so we can reduce the statement to show that averaging over the Haar measure of the direct sum of unitaries leads to a coarse-grained state, i.e.,
\begin{equation}
\int (\bigoplus_i{U}_i)\R (\bigoplus_i{U}_i^\dag) d\mu(\bigoplus_i{U}_i)=\R_\cg.
\end{equation}
We also dropped the tilde to increase clarity.

We do this in several stages by proving two Lemmas first.

\begin{lemma}
Let $\ket{\psi}$ be a pure state. Let $d\mu(U)$ be the Haar measure over unitary group on the full Hilbert space. Then
\[
\int U\pro{\psi}{\psi}U^\dag d\mu(U)=I/d,
\]
where $I/d$ is the maximally mixed state, $I$ is the identity operator and $d$ is the dimension of the Hilbert space.
\end{lemma}
\begin{proof}
Let $U'$ be a fixed unitary of our choice. We have
\[
\begin{split}
\int U\pro{\psi}{\psi}U^\dag d\mu(U)&=\int \tilde U U'\pro{\psi}{\psi} U'^\dag \tilde U^\dag d\mu(\tilde U U')\\
&=\int \tilde U \pro{\psi'}{\psi'} \tilde U^\dag d\mu(\tilde U U')\\
&=\int \tilde U \pro{\psi'}{\psi'} \tilde U^\dag d\mu(\tilde U)\\
&=\int U \pro{\psi'}{\psi'} U^\dag d\mu(U)\\
\end{split}
\]
where we have defined $\tilde U= U U'^\dag$, and $\ket{\psi'}=U'\ket{\psi}$. In the third equality, we used the right invariance of the Haar measure (which we can do because the unitary group is unimodular and thus its Haar measure is both left- and right-invariant). In the fourth equality, we relabeled $\tilde U$ as $U$. $U'$ was arbitrary, and any vector of the Hilbert space can be transformed to any other vector in a Hilbert space with a unitary operator, which means $\ket{\psi'}$ can be absolutely any pure state. Let us choose $\ket{\psi'}$ to be orthogonal to $\ket{\psi}$. We repeat this until we form a basis of $d$ vectors $\{\ket{\psi_i}\}_{i=1}^d$ ($\ket{\psi}=\ket{\psi_1}$, $\ket{\psi'}=\ket{\psi_2}$ and so on.) From the above, we know that all of these averages are equal to the same operator,
\[
\int U \pro{\psi_1}{\psi_1} U^\dag d\mu(U)=\cdots=\int U \pro{\psi_d}{\psi_d} U^\dag d\mu(U)=:O.
\]
This means that we can write
\[
\begin{split}
O&=\frac{1}{d}\sum_{i=1}^d\int U \pro{\psi_i}{\psi_i} U^\dag d\mu(U)\\
&=\frac{1}{d}\int U \sum_{i=1}^d\pro{\psi_i}{\psi_i} U^\dag d\mu(U)\\
&=\frac{1}{d}\int U I U^\dag d\mu(U)\\
&=\frac{1}{d}\int I d\mu(U)\\
&=I/d\\
\end{split}
\]
where we used the fact that for an orthonormal basis, $I=\sum_{i=1}^d\pro{\psi_i}{\psi_i}$.
\end{proof}

\begin{lemma}\label{lem:mixing}
For any state $\R$, we have
\[
\int U\R U^\dag d\mu(U)=I/d.
\]
\end{lemma}
\begin{proof}
Every state has a spectral decomposition $\R=\sum_i\lambda_i\pro{\psi_i}{\psi_i}$. We have
\[
\begin{split}
\int U\R U^\dag d\mu(U)&=\sum_i\lambda_i\int U\pro{\psi_i}{\psi_i} U^\dag d\mu(U)\\
&=\sum_i\lambda_i I/d=I/d.
\end{split}
\]
\end{proof}

Finally, we obtain the statement of the theorem
\begin{theorem}
We have
\[
\int (U_1\oplus U_2)\R (U_1^\dag\oplus U_2^\dag) d\mu(U_1\oplus U_2)=\frac{p_1}{V_1}\P_2+\frac{p_1}{V_1}\P_2.
\]
\end{theorem}
\begin{proof}
Before moving to the actual proof of the final statement, let us illustrate the meaning of 
\[
(U_1\oplus U_2)\R (U_1^\dag\oplus U_2^\dag)
\]
on its corresponding matrix representation, which will make the rest of the proof clearer. We have block-matrix representations
\[
U_1\oplus U_2=
\begin{pmatrix}
    U_1 & 0 \\
    0 & U_2
\end{pmatrix},\quad U_1^\dag\oplus U_2^\dag=
\begin{pmatrix}
    U_1^\dag & 0 \\
    0 & U_2^\dag
\end{pmatrix}
\]
and 
\[
\R=
\begin{pmatrix}
    \R_{11} & \R_{12} \\
    \R_{21} & \R_{22}
\end{pmatrix}.
\]
Thus, in the block-matrix representation, we have 
\[\label{eq:direct_sum_on_matrix}
\begin{split}
(U_1\oplus U_2)\R (U_1^\dag\oplus U_2^\dag)=\begin{pmatrix}
    U_1\R_{11}U_1^\dag & U_1\R_{12}U_2^\dag \\
    U_2\R_{21}U_1^\dag & U_2\R_{22}U_2^\dag
\end{pmatrix}.
\end{split}
\]
$\R_{11}$ is the part of the density matrix that lives in the subspace $P_1\HS P_1$. To be more precise, it is isomorphic to the projection of the density matrix onto this subspace, $P_1\R P_1$. The operator $P_1\R P_1$ is technically an operator on the entire Hilbert space, but with its support only in the subspace. On the full Hilbert space, the operator $P_1\R P_1$ has a block-matrix representation
\[
P_1\R P_1=\begin{pmatrix}
    \R_{11} & 0 \\
    0 & 0
\end{pmatrix}
\]
Similarly, the $\R_{22}$ is isomorphic to $P_2\R P_2$, and then non-diagonal block terms: $\R_{12}$ is isomorphic to $P_1\R P_2$ and $\R_{21}$ is isomorphic to $P_2\R P_1$. 

Now we move to the actual proof which is done in the operator form. We can write
\[
\R=P_1\R P_1 + P_1\R P_2+P_2\R P_1+P_2\R P_2,
\]
which follows from the completeness relation $P_1+P_2=I$.

In the operator representation,  Eq.~\eqref{eq:direct_sum_on_matrix} rewrites as
\[
\begin{split}
   &(U_1\oplus U_2)\R (U_1^\dag\oplus U_2^\dag)=\\
   &U_1P_1\R P_1U_1^\dag + U_1P_1\R P_2U_2^\dag+U_2P_2\R P_1U_1^\dag+U_2P_2\R P_2U_2^\dag.
\end{split}
\]
In the above, we used a shorthand $U_1\equiv U_1\oplus 0$ and $U_2\equiv 0\oplus U_2$, since $U_1$ and $U_2$ technically act only at subspaces one and two, respectively.

The Haar measure factorizes:
\[
d\mu(U_1\oplus U_2)=d\mu(U_1)d\mu(U_2),
\]
where $d\mu(U_1)$ and $d\mu(U_2)$ are the normalized Haar measures on unitary operators applied on subspaces. This is justified as follows: The Haar measure considered here is that defined on the group of unitary operators that can be written as a direct sum of unitary operators applied on the subspaces. Defining a measure on this group by the right-hand side of the above equation, we can show that all of the properties of the Haar measure are satisfied. The Haar measure is unique up to a multiplicative constant. Here we consider the normalized Haar measure, which means that it is unique. Thus, the measure defined by the right-hand side must be the Haar measure and the equation holds.

Therefore, taking the averaging, we obtain a sum of four terms:
\[
\begin{split}
   &\int (U_1\oplus U_2)\R (U_1^\dag\oplus U_2^\dag) d\mu(U_1\oplus U_2)\\
   &= \int U_1P_1\R P_1U_1^\dag d\mu(U_1)d\mu(U_2)\\
   &+\int U_1P_1\R P_2U_2^\dag d\mu(U_1)d\mu(U_2)\\
   &+\int U_2P_2\R P_1U_1^\dag d\mu(U_1)d\mu(U_2)\\
   &+\int U_2P_2\R P_2U_2^\dag d\mu(U_1)d\mu(U_2).
\end{split}
\]

Taking the first term, we have 
\[
\begin{split}
    \int U_1P_1&\R P_1U_1^\dag d\mu(U_1)d\mu(U_2)=\int U_1P_1\R P_1U_1^\dag d\mu(U_1)\\
    &=p_1\int U_1\frac{P_1\R P_1}{p_1}U_1^\dag d\mu(U_1)=p_1 \frac{P_1}{V_1},
\end{split}
\]
where $p_1=\tr[P_1\R]$ is the normalization factor. In the last equality, we applied Lemma~\ref{lem:mixing} on matrix $\frac{P_1\R P_1}{p_1}$, where we have used that $P_1$ is the identity matrix on the subspace $P_1\HS P_1$, and $V_1=\tr[P_1]$ is the dimension of this subspace.

The second term is zero. We prove this by a similar method to the proof of Lemma~\ref{lem:mixing}. We use that the minus identity, $-I$, is a unitary operator. We have
\[
\begin{split}
    \int U_1P_1&\R P_2U_2^\dag d\mu(U_1)d\mu(U_2)\\
    &=\int U_1P_1\R P_2(\tilde U_2^\dag(-I) ) d\mu(U_1)d\mu((-I) \tilde U_2)\\
    &=\int U_1P_1\R P_2(\tilde U_2^\dag(-I)) d\mu(U_1)d\mu(\tilde U_2)\\
    &=-\int U_1P_1\R P_2U_2^\dag d\mu(U_1)d\mu(U_2).
\end{split}
\]
We have defined $U_2=(-I) \tilde U_2$. In the second equality, we used the invariance of the Haar measure. In the last equality, we relabeled $\tilde U_2$ as $U_2$. We showed that the second term is equal to its negative, therefore, it must be zero.

The last two terms proceed analogously, and we obtain 
\[
\int (U_1\oplus U_2)\R (U_1^\dag\oplus U_2^\dag) d\mu(U_1\oplus U_2)=\frac{p_1}{V_1}P_1+\frac{p_2}{V_2}P_2.
\]
\end{proof}
Clearly, the theorem generalizes to any number of projectors.

\section{Derivation of the limit state in Stage 1.}\label{sec:2}

In this section, we show that the limiting state after $N$ measurements in Stage 1.~is given by
\begin{equation}\label{eq:limit_state_stage1}
\R_N=\bigotimes_i\R_i^{\otimes p_i N}\equiv \R^N
\end{equation}
up to a permutation in the order of tensor products, as $N$ to infinity. This is Eq.~(25) 
in the main text. The left hand side is a random variable (random state) $\R_N=\R_{i_1}\otimes\cdots\otimes\R_{i_N}$ which is produced with probability $p_{i_1\dots i_N}=p_{i_1}\dots p_{i_N}$ after measuring each copy produced by the source (see the proof below for details). The right hand side is defined by Eq.~\eqref{eq:limit_state_floor}. In mathematical terms, we show that $\R_N$ converges to $\bigotimes_i\R_i^{\otimes p_i N}$ \emph{almost surely}. To show that, we first need to introduce relevant terminology and results.

For independent and identical distributed (i.i.d.) source $X$ (meaning that all $X_i=X$ are the same random variables), the Asymptotic Equipartition Property states that for any $\epsilon>0$,
\begin{equation}\label{eq:AEP}
\lim_{N\rightarrow \infty} \mathrm{Pr}\left[\abs{-\frac{1}{N}\ln p(X_1,\dots,X_N)-H(X)}>\epsilon\right]=0.
\end{equation}
where $H(X)=-\sum_Xp(X)\ln p(X)$ is the Shannon entropy of the outcomes, and $p_i$ probability of event $X_i$. (See, e.g.~\cite{cover1999elements}.)
In other words, for any fixed $\epsilon$, probability that we find sequence $(X_1,\dots,X_N)$ to have probability $p(X_1,\dots,X_N)$ outside these bounds:
\begin{equation}\label{eq:typical_set}
e^{-N(H(X)+\epsilon)}\leq p(X_1,\dots,X_N)\leq e^{-N(H(X)-\epsilon)}
\end{equation}
is zero as $N$ goes to infinity., i.e., vast majority of probabilities will fall within these bounds for large $N$. Eq.~\eqref{eq:typical_set} is the defining property of the typical set: typical set is defined as $B_N^\epsilon=\{(X_1,\dots,X_N)|\, \mathrm{Eq.}~\eqref{eq:typical_set}\text{ holds}\}$.

For i.i.d.~sources we actually have a stronger property: almost sure convergence,
\begin{equation}\label{eq:almost_sure}
\mathrm{Pr}\left[\lim_{N\rightarrow\infty}-\frac{1}{N}\ln p(X_1,\dots,X_N)=H(X)\right]=1,
\end{equation}
which is implied by the strong law of large numbers. This states that any sequence of events $(X_1,\dots,X_N)$ for which
\begin{equation}
p(X_1,\dots,X_N)\neq e^{-N H(X)}
\end{equation}
have probability zero, as $N$ goes to infinity.

We can now proceed with the proof.

\begin{proof}
We will show that in the limit of large $N$, the resulting state is almost surely of form~\eqref{eq:limit_state_floor} up to a different order of tensor products (i.e., up to a permutation).

For a single measurement, the resulting state is $\R_i$ with probability $p_i\equiv p(\R_i)$, for two measurements, the resulting state is $\R_{i_1}\otimes \R_{i_2}$ with probability $p_{i_1i_2}=p_{i_1}p_{i_2}$, for $N$ measurements, it is $\R_{i_1}\otimes \cdots\otimes \R_{i_N}$ with probability $p_{i_1\dots i_N}=p_{i_1}\cdots p_{i_N}$. The sequence $\R_N=\R_{i_1}\otimes \cdots\otimes \R_{i_N}$ is therefore generated by an i.i.d. process in which the next element is given by $\R_i$ with probability $p_i$. We define a set
\begin{equation}
A_N=\{\pi(\R_1^{\otimes(N-\sum_{i\geq 2}\lfloor p_i N \rfloor)}\bigotimes_{i\geq 2}\R_i^{\otimes \lfloor p_i N \rfloor})\}_\pi,
\end{equation}
where $\pi$ defines all permutations on the tensor product. The probability of any state from this set is 
\begin{equation}
p_{i_1\dots i_N}=p_1^{N-\sum_{i\geq 2}\lfloor p_i N \rfloor}\prod_{i\geq 2}p_i^{\lfloor p_i N \rfloor}.
\end{equation}
Taking the logarithm we have
\begin{equation}
\begin{split}
-\frac{1}{N}\ln p_{i_1\dots i_N}&=-\frac{1}{N}\left(\Big(N\!\!-\!\!\sum_{i\geq 2}\lfloor p_i N \rfloor\Big)\ln p_1+\sum_{i\geq 2} {\lfloor p_i N \rfloor} \ln p_i\right)\\
&=-\sum_i p_i\ln p_i+\mathcal{O}(\tfrac{1}{N})=H(\R)+\mathcal{O}(\tfrac{1}{N}).
\end{split}
\end{equation}
Since $\mathcal{O}(\tfrac{1}{N})$ will be smaller than any $\epsilon$, Eq.~\eqref{eq:AEP} is satisfied for every element of set $A_N$. Further, due to the number of permutations $\pi$ that generate a different sequence, the number of elements in the set is given by a multinomial,
\begin{equation}
\abs{A_N}=\frac{N!}{(N-\sum_{i\geq 2}\lfloor p_i N \rfloor)!\lfloor p_2 N \rfloor!\cdots \lfloor p_k N \rfloor!}
\end{equation}
Taking the logarithm and using Stirling's approximation $\ln N!=n\ln N+N+\frac{1}{2}\ln (2\pi N)$ we derive
\begin{equation}
\ln\abs{A_N}=N\left(H(\R)+\mathcal{O}\left(\tfrac{\ln(N)}{N}\right)\right).
\end{equation}
Again, $\mathcal{O}\left(\tfrac{\ln(N)}{N}\right)$ will be smaller than any $\epsilon$ for large enough $N$. This gives,
\begin{equation}
e^{N(H(\R)-\epsilon)}\leq\abs{A_N}\leq e^{N(H(\R)+\epsilon)},
\end{equation}
where $H(\R)=-\sum_i p_i\ln p_i$. This means that set $A_N$ contains elements with the same probability as the typical set 
\begin{equation}
B_N^\epsilon\!=\!\left\{\R_{i_1}\!\otimes\! \cdots \!\otimes\! \R_{i_N}\lvert e^{-N(H(\R)+\epsilon)}\leq p_{i_1\dots i_N}\leq e^{-N(H(\R)-\epsilon)}\right\}
\end{equation}
does, and its size (cardinality) is also the same as the typical set. This means that there will be a large overlap between the two and that they will share the same properties. Namely, for any arbitrarily small $\delta>0$, the probability of finding the randomly generated sequence $\R_{i_1}\!\otimes\! \cdots \!\otimes\! \R_{i_N}$ in $A_N$ is $\Pr{A_N}\geq 1-\delta$, if $N$ is taken to be large enough~\cite{cover1999elements}. Even the stronger condition of almost sure convergence, Eq.~\eqref{eq:almost_sure}, holds, meaning that for a large enough $N$, sequences that do not fall into set $A_N$ have probability zero. This completes the proof that for a large $N$, all states have the form \eqref{eq:limit_state_floor} (or Eq.~\eqref{eq:limit_state_stage1} if we gloss over the non-integer exponents), up to a permutation in the order of tensor products.
\end{proof}

\section{Boltzmann ergotropy in the large $N$ limit}\label{sec:3}

In this section, we combine  Eqs.~\eqref{eq:original_sequnce_proof_app} and~\eqref{eq:limit_state_stage1} (corresponding to Eqs.~(20) 
and~(25) 
in the main text) to derive expression for the Boltzmann ergotropy in the limit of large $N$, 
\begin{equation}
W_\C^{B\infty}(\R):=\lim_{N\rightarrow \infty}\frac{W_\C^B(\R^{\otimes N})}{N}=\tr[\ham(\R-\R_{\beta})],
\end{equation}
where temperature $\beta$ is implicitly defined by $S_{\vN}(\R_\beta)=S_\C^{B}:=\sum_i p_i\ln V_i$.

To clarify, note that $\R^{\otimes N}$ above denotes the initial state produced by the source, i.e., the source produced $N$ identical copies of the same unknown state $\R$. Compare it to a random variable (random density matrix) $\R_N=\R_{i_1}\otimes\cdots\otimes\R_{i_N}$ which is produced with probability $p_{i_1\dots i_N}=p_{i_1}\dots p_{i_N}$ after measuring each copy produced by the source, and the limiting state $\R^N=\bigotimes_i\R_i^{\otimes p_i N}$, defined by Eq.~\eqref{eq:limit_state_floor}, to which $\R_N$ almost surely converges.

Recall the definition of the multipartite unitary extraction when averaged over many realizations of the protocol, Eq.~(9) 
in the main text:
\begin{equation}
\label{eq:global_extraction_operation_multiple_app}
\mathcal{U}(\R^{(1)}\otimes\cdots \otimes\R^{(N)})=U\R_\cg^{(1)}\otimes\cdots \otimes\R_\cg^{(N)}U^\dag,
\end{equation}
Using the same logic as in Eq.~(30) 
in the main text, 
we apply $\mathcal{U}$ to $\R_N$ to obtain
\begin{equation}\label{eq:WcB}
\begin{split}
&W_\C^{B}(\R^{\otimes N})=\tr[\ham_N\R^{\otimes N}]- \min_U\tr[\ham_N\mathcal{U}\big(\R_N\big)]\\
&=\tr[\ham_N\R^{\otimes N}]- \min_U\tr[\ham_N\mathcal{U}\big(\R^N\big)]\\
&=N\tr[\ham\R]-\min_U\tr[\ham_N U\Big(\bigotimes_i(\P_i/V_i)^{\otimes p_i N}\Big)U^\dag].
\end{split}
\end{equation}
where we used Eq.~\eqref{eq:limit_state_stage1} in the second equality, assuming that $N$ is large.
Dividing this equation by $N$ and using Eq.~\eqref{eq:original_sequnce_proof_app}, we have
\begin{equation}
\lim_{N\rightarrow \infty}\frac{W_\C^B(\R^{\otimes N})}{N}=\tr[\ham\R]-\tr[\ham\R_\beta],
\end{equation}
where
\begin{equation}
S_{\vN}(\R_{\beta})=\sum_i p_i S_{\vN}(\P_i/V_i)=\sum_i p_i \ln V_i =:S_\C^{B}.
\end{equation}
This concludes the proof.

\begin{figure}[t]
\begin{center}
\includegraphics[width=1\hsize]{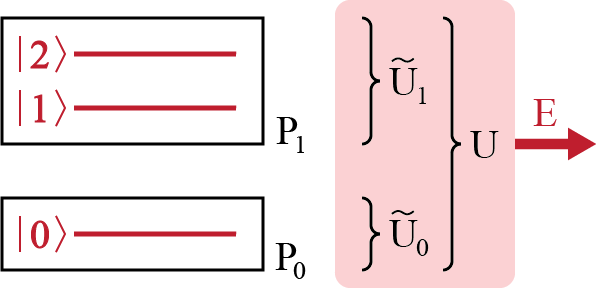}\\
\caption
{Example of a three-level system with two-outcome coarse-graining. $\P_0$ and $\P_1$ represents two macrostates, $\ket{0}$, $\ket{1}$, $\ket{2}$ energy eigenstates, $\tilde{U}_0$ and $\tilde{U}_1$ random unitaries acting on the first and the second subspace-macrostate, respectively, and $U$ is the global extraction unitary.}
\label{Fig:threelevel}
\end{center}
\end{figure}

\section{Simple analytic example}\label{sec:4}

Here we illustrate the coarse-grained extraction unitary,
\begin{equation}
U\bigoplus_i \tilde U_i,
\end{equation}
where $U_i$ are random unitary operators picked with the Haar measure acting on macrostates, and $U$ is the global unitary, on an example depicted in Fig~\ref{Fig:threelevel}. This is Eq.~(5) 
in the main text. Then we compute corresponding Boltzmann and observational ergotropy for several sources of unknown states.

Consider a three-level system described by Hamiltonian
\begin{equation}
\ham=E_0\pro{0}{0}+E_1\pro{1}{1}+E_2\pro{2}{2}.
\end{equation}
We form a two-outcome coarse-graining by grouping together the two excited states, while the ground state will form its own macrostate. Mathematically, we have
\begin{equation}
\C=\{\P_0, P_{1}\}
\end{equation}
as our coarse-graining, where
\begin{equation}\label{eq:cg_two_outcomes}
\P_0=\pro{0}{0},\quad \P_{1}=\pro{1}{1}+\pro{2}{2}
\end{equation}
are the corresponding projectors. The corresponding subspaces-macrostates are
\begin{equation}
\HS_0=\mathrm{span}\{\ket{0}\},\quad \HS_1=\mathrm{span}\{\ket{1},\ket{2}\}.
\end{equation}
Clearly, we have $\HS=\HS_0\oplus \HS_1$.

The random unitaries are of the form
\begin{equation}
\tilde U_0\oplus \tilde U_1=
\begin{pmatrix}
\tilde U_0 & \begin{matrix} 0 & 0 \end{matrix} \\
\begin{matrix} 0 \\ 0 \end{matrix} & \tilde  U_1
\end{pmatrix},
\end{equation}
where $\tilde  U_0$ is a $1\times 1$ random unitary matrix  and $\tilde U_1$ is a $2\times 2$ random unitary matrix. (For an efficient sampling method of unitary operators from the Haar measure see, e.g., Ref.~\cite{lundberg2004haar}.) The global extraction unitary $U$ is a general, $3 \times 3$ unitary matrix.

\subsection{Three pure states that are indistinguishable by the measurement}

We illustrate the difference the effects of the coarse-grained extraction unitary and a generic extraction unitary on several different states that could be produced by an unknown source.

Consider three states:
\begin{equation}
\R_A=\pro{1}{1},\quad \R_B=\pro{2}{2},\quad \R_C=\pro{+}{+},
\end{equation}
where $\ket{+}=(\ket{1}+\ket{2})/\sqrt{2}$. All of these have the same probabilities of outcomes when performing a measurement given by the coarse-graining~\eqref{eq:cg_two_outcomes},
\begin{equation}
p_{0}=0,\quad p_1=1.
\end{equation}
Because the measurement outcome is $i=1$ with certainty, in all cases the post-measurement state is the same as the initial state, $\R_1^X=\R_X$ ($X=A,B,C$). It also means that the corresponding coarse-grained density matrix obtained by averaging over the random unitaries picked with the Haar measure,
\begin{equation}
\begin{split}
  \R_{\cg}&=\int \big(\tilde U_0\oplus \tilde U_1\big) \R_X \big(\tilde U_0^\dag\oplus \tilde U_1^\dag\big) d\mu(\tilde U_0\oplus \tilde U_1)\\
&=\frac{p_0}{V_0}\P_0+\frac{p_1}{V_1}\P_1\\
&=0\begin{pmatrix}
1 & 0 & 0 \\
0 & 0 & 0 \\
0 & 0 & 0 
\end{pmatrix}+\frac{1}{2}
\begin{pmatrix}
0 & 0 & 0 \\
0 & 1 & 0 \\
0 & 0 & 1 
\end{pmatrix}=\begin{pmatrix}
0 & 0 & 0 \\
0 & \tfrac{1}{2} & 0 \\
0 & 0 & \tfrac{1}{2} 
\end{pmatrix},
\end{split}
\end{equation}
is the same of all of them (where $X=A,B,C$, $V_0=1$, $V_1=2$). The optimal work extraction unitary from this state is given by
\begin{equation}\label{eq:U_extraction}
U=\pro{0}{2}+\pro{2}{0}+\pro{1}{1}=\begin{pmatrix}
0 & 0 & 1 \\
0 & 1 & 0 \\
1 & 0 & 0 
\end{pmatrix}.
\end{equation}
This unitary turns the coarse-grained state into a passive state,
\begin{equation}
\pi_\cg=U\R_\cg U^\dag=\mathcal{U}(\R_X)=\begin{pmatrix}
\tfrac{1}{2} & 0 & 0 \\
0 & \tfrac{1}{2} & 0 \\
0 & 0 & 0 
\end{pmatrix}.
\end{equation}
The initial energies of these three states are different (and unknown to the experimenter), and given by
\begin{equation}
\begin{split}
    \tr[\ham \R_A]&=E_1,\\
    \tr[\ham \R_B]&=E_2,\\
    \tr[\ham \R_C]&=(E_1+E_2)/2.
\end{split}
\end{equation}
However, in all of these cases the energy of the respective state is reduced to the same value through the optimal work extraction unitary $U$, i.e.,
\begin{equation}
\tr[\ham \pi_\cg]=(E_0+E_1)/2,
\end{equation}
where $X=A,B,C$.

The corresponding Boltzmann ergotropy is given by
\begin{equation}
W_\C^{B}(\R_X)=\mathrm{tr}[\ham(\R_X-\pi_\cg)].
\end{equation}
This yields
\begin{equation}
\begin{split}
    W_\C^{B}(\R_A)&=(E_1-E_0)/2,\\
    W_\C^{B}(\R_B)&=E_2-(E_0+E_1)/2,\\
    W_\C^{B}(\R_C)&=(E_2-E_0)/2.\\
\end{split}
\end{equation}
Note that the observational ergotropy is the same as the Boltzmann ergotropy in this case. This is because $p_0=0$ and $p_1=1$.

\subsection{Comparison with blind direct extraction}

Let us compare the above strategy to the blind strategy of extraction, without the averaging using random unitaries picked with the Haar measure. The experimenter does not know the actual state of the system; their only knowledge is of the probabilities $p_0=0$ and $p_1=1$. Let us say that they decide to extract energy directly by performing a unitary operation $U$ which transfers the energy from the second excited state to the ground state, given by Eq.~\eqref{eq:U_extraction}.
In this case the final state energy is
\begin{equation}
\tr[\ham \R_X^{\mathrm{final}}]=\tr[\ham\, U\R_X U^\dag].
\end{equation}
This gives
\begin{equation}
\begin{split}
    \tr[\ham \R_A^{\mathrm{final}}]&=E_1\\
    \tr[\ham \R_B^{\mathrm{final}}]&=E_0,\\
    \tr[\ham \R_C^{\mathrm{final}}]&=(E_0+E_1)/2.\\
\end{split}
\end{equation}
The total extracted work is then given by
\begin{equation}
W (\R_X)=\tr\big[\ham \big(\R_X-\R_X^{\mathrm{final}}\big)\big].
\end{equation}
This gives
\begin{equation}
\begin{split}
    W (\R_A)&=0\\
    W (\R_B)&=E_2-E_0,\\
    W (\R_C)&=(E_2-E_0)/2.\\
\end{split}
\end{equation}
Clearly, this makes a very unreliable source, strongly depending on the unknown state of the system and on how lucky is the experimenter in guessing which extraction unitary $U$ to apply. While in some cases, the extracted energy is quite high due to being lucky, in some other cases the strategy fails.

\subsection{What if the initial state was known?}

Since all the states $A,B,C$ are pure, if the experimenter knew what the states are, they would be able to extract all the mean energy by transferring the state into the ground state. The extracted work would be given by the usual notion of ergotropy, which leads to
\begin{equation}
W_X=\tr[\ham \R_X]-E_0.
\end{equation}
This is, however, not the case in our setup.

\subsection{Another example}
\subsubsection{Boltzmann ergotropy vs blind extraction}
Finally, we show an example that can be very unfavorable to the blind direct extraction. Consider a state
\begin{equation}
\R_D=\frac{1}{8}\pro{0}{0}+\frac{7}{8}\pro{1}{1}=\begin{pmatrix}
\tfrac{1}{8} & 0 & 0 \\
0 & \tfrac{7}{8} & 0 \\
0 & 0 & 0
\end{pmatrix},
\end{equation}
which has the mean initial energy
\begin{equation}
\tr[\ham\, \R_D]=\tfrac{1}{8}E_0+\tfrac{7}{8}E_1.
\end{equation}
We have $p_0=\frac{1}{8}$ and $p_1=\frac{7}{8}$. 

If the outcome of the measurement is $i=1$ we obtain the projected state $\R_1=\pro{1}{1}$ after the measurement. The corresponding coarse-grained state is given by
\begin{equation}
\R_\cg^1=\P_1/V_1=\begin{pmatrix}
0 & 0 & 0 \\
0 & \tfrac{1}{2} & 0 \\
0 & 0 & \tfrac{1}{2} 
\end{pmatrix}.
\end{equation}
The optimal extraction unitary is again Eq.~\eqref{eq:U_extraction}, which turns this coarse-grained state into a passive state
\begin{equation}
\pi_1=\begin{pmatrix}
\tfrac{1}{2} & 0 & 0 \\
0 & \tfrac{1}{2} & 0 \\
0 & 0 &  0
\end{pmatrix}.
\end{equation}
If the outcome is $i=0$ we obtain the projected state $\R_0=\pro{0}{0}$ after the measurement. The corresponding coarse-grained state is already passive and equal to the projected state,
\begin{equation}
\R_\cg^0=\pi_0=\R_0.
\end{equation}
As a result, no energy can be extracted in this case.

The total extractable work is given by Boltzmann ergotropy,
\begin{equation}
\begin{split}
W_\C^{B}(\R_D)&=\mathrm{tr}[\ham(\R_D-p_0\pi_0-p_1\pi_1)]\\
&=\tfrac{1}{8}E_0+\tfrac{7}{8}E_1-\tfrac{1}{8}E_0-\tfrac{7}{16}(E_0+E_1)\\
&=\tfrac{7}{16}(E_1-E_0).\\
\end{split}
\end{equation}

Compare this to the blind extraction protocol, assuming that the experimenter chooses unitary operator Eq.~\eqref{eq:U_extraction} to extract the energy directly, in the case when they receive outcome $i=1$ from the measurement. In the case they receive outcome $i=0$, they decide to do nothing, because they already know from the outcome that the state is passive. Thus, the final states are
\begin{equation}
\R_1^{\mathrm{final}}=U \pro{1}{1} U^\dag=\pro{1}{1},\quad \R_0^{\mathrm{final}}=\pro{0}{0}.
\end{equation}
The average extracted work is
\begin{equation}
W(\R_D)=\tr[\ham \R_D]-p_0\tr[\ham \R_0^{\mathrm{final}}]-p_1\tr[\ham \R_1^{\mathrm{final}}]=0.
\end{equation}
Thus, blind extraction does not extract anything in this case, while the coarse-grained extraction does.

\subsubsection{Observational ergotropy vs blind extraction}

Now consider the experimenter moved to stage 2., in which the probabilities $p_i$ are already identified, and no further measurements will be performed.

The coarse-grained state corresponding to the initial state $\R_D$ is given by
\begin{equation}
\R_\cg=\frac{1}{8}\P_0+\frac{7}{16}\P_1=\begin{pmatrix}
\tfrac{1}{8} & 0 & 0 \\
0 & \tfrac{7}{16} & 0 \\
0 & 0 & \tfrac{7}{16} 
\end{pmatrix}.
\end{equation}
The optimal extraction unitary, which is again of form~\eqref{eq:U_extraction}, will turn this into a passive state
\begin{equation}
\pi_\cg=\begin{pmatrix}
\tfrac{7}{16} & 0 & 0 \\
0 & \tfrac{7}{16} & 0 \\
0 & 0 &  \tfrac{1}{8}
\end{pmatrix}.
\end{equation}
The extracted work is given by observational ergotropy,
\begin{equation}
\begin{split}
   W_\C(\R_D)&=\tr[\ham(\R_D-\pi_\cg)]\\
   &=\tfrac{1}{8}E_0+\tfrac{7}{8}E_1-\tfrac{7}{16}E_0-\tfrac{7}{16}E_1-\tfrac{1}{8}E_2\\
   &= (7 E_1-5 E_0-2 E_2)/16.
\end{split}
\end{equation}
The extracted work is positive if $7 E_1> 5 E_0+2 E_2$. It is negative in the opposite case, signifying that the experimenter would lose energy. Of course, the computation crucially depends on the unknown initial energy, but as we mentioned in the main text, this can be estimated, either from Eq.~(17) 
in the main text in the case of local energy coarse-grainings, or using methods in~\cite{safranek2023expectation} in the case of general coarse-grainings.

We compare this to the blind extraction, in which the unitary~\eqref{eq:U_extraction} is applied directly on the state produced by the source. In this case, the final state is
\begin{equation}
\R_D^{\mathrm{final}}=U \R_D U^\dag=\begin{pmatrix}
0 & 0 & 0 \\
0 & \tfrac{7}{8} & 0 \\
0 & 0 & \tfrac{1}{8} 
\end{pmatrix},
\end{equation}
and the extracted work is given by
\begin{equation}
W(\R_D)=\tr[\ham (\R_D-\R_D^{\mathrm{final}})]=(E_0-E_2)/8.
\end{equation}
This is always negative, so the experimenter is bound to lose energy.

\section{Bound on energy for local energy coarse-grainings}\label{sec:5}

Here we prove the bound on energy that can be obtained when knowing the probabilities of outcomes of local energy measurements,
\begin{equation}
\abs{\tr[\ham \R]-\tr[\ham \R_\cg]}\leq 2\norm{H_{\mathrm{int}}}+k\Delta E,
\end{equation}
which is Eq.~(17) 
in the main text.
$\R_\cg=\sum_{E_1,E_2}\frac{p_{E_1E_2}}{V_{E_1}V_{E_1}}\P_{E_1}\otimes \P_{E_2}$ is the coarse-grained state given by the local energy coarse-grainings $\C=\{\P_{E_1}\otimes \P_{E_2}\}$. $\P_{E_i}=\sum_{E\in[E_i,E_i+\Delta E)}\pro{E_i}{E_i}$ are the coarse-grained projectors on local energies with resolution $\Delta E$, and 
$\ham=\ham_1+\ham_2+H_{\mathrm{int}}$, where spectral decompositions of the local Hamiltonians are $\ham_1=\sum_{E_1}E_1\pro{E_1}{E_1}$ and $\ham_2=\sum_{E_2}E_2\pro{E_2}{E_2}$. 

We have
\begin{widetext}
\begin{equation}
\begin{split}
\abs{\tr[\ham \R]-\tr[\ham \R_\cg]}&=\abs{\tr[(\ham_1+\ham_2)(\R-\R_\cg)]+\tr[H_{\mathrm{int}} (\R-\R_\cg)]}\\
&\leq \abs{\tr[\Big(\sum_{E_1,E_2}(E_1+E_2)\pro{E_1}{E_1}\otimes\pro{E_2}{E_2}\Big)(\R-\R_\cg)]}+\abs{\tr[H_{\mathrm{int}} (\R-\R_\cg)]}\\
&\leq \left|\sum_{E_1,E_2}(E_1+E_2)\bra{E_1,E_2}(\R-\R_\cg)\ket{E_1,E_2}\right|+2\norm{H_{\mathrm{int}}}\\
&\leq \left|\sum_{E_1',E_2'}\sum_{E_1\in [E_1',E_1'+\Delta E),E_2\in [E_2',E_2'+\Delta E)}(E_1+E_2)\bra{E_1,E_2}(\R-\R_\cg)\ket{E_1,E_2}\right|+2\norm{H_{\mathrm{int}}}\\
&\leq \bigg|\sum_{E_1',E_2'}(E_1'+E_2')\sum_{E_1\in [E_1',E_1'+\Delta E),E_2\in [E_2',E_2'+\Delta E)}\bra{E_1,E_2}(\R-\R_\cg)\ket{E_1,E_2}\\
&+\sum_{E_1',E_2'}\sum_{E_1\in [E_1',E_1'+\Delta E),E_2\in [E_2',E_2'+\Delta E)}(E_1+E_2-E_1'-E_2')\bra{E_1,E_2}\R\ket{E_1,E_2}\bigg|+2\norm{H_{\mathrm{int}}}\\
&\leq 0+\sum_{E_1,E_2}2\Delta E\, p_{E_1E_2}+2\norm{H_{\mathrm{int}}}
=2\Delta E+2\norm{H_{\mathrm{int}}}.
\end{split}
\end{equation}
\end{widetext}
In the above, $E_1,E_2$ denote fine-grained energies, and $E_1',E_2'$ coarse-grained energies. Generalization for $k$ partitions is $\abs{\tr[\ham \R]-\tr[\ham \R_\cg]}\leq k\Delta E+2\norm{H_{\mathrm{int}}}$. $\norm{H_{\mathrm{int}}}$ is also linearly proportional to $k$, therefore, both terms represent a finite size effect.

\end{document}